\documentclass[12pt]{article}%
\usepackage{amsfonts}
\usepackage{amsmath}
\usepackage{amssymb}
\usepackage{graphicx}%
\providecommand{\U}[1]{\protect\rule{.1in}{.1in}}
\providecommand{\U}[1]{\protect\rule{.1in}{.1in}}
\newtheorem{theorem}{Theorem}

\newtheorem{lemma}[theorem]{Lemma}

\newenvironment{proof}[1][Proof]{\noindent\textbf{#1.} }{\ \rule{0.5em}{0.5em}}
\begin{document}

\title{A quantum polylog algorithm for non-normal maximal cyclic hidden subgroups in
the affine group of a finite field}
\author{Nolan Wallach\\University of California, San Diego\\nwallach@ucsd.edu}
\maketitle

\begin{abstract}
We give an algorithm to solve the quantum hidden subgroup problem for maximal
cyclic non-normal subgroups of the affine group of a finite field (if the
field has order $q$ then the group has order $q(q-1)$) with probability
$1-\varepsilon$ with (polylog) complexity $O(\log(q)^{R}\log(\varepsilon
)^{2})$ where $R<\infty.$

\end{abstract}

\section{Introduction}

The purpose of this paper is to describe and prove an algorithm that solves
the quantum hidden subgroup problem for maximal cyclic non-normal subgroups of
the affine group of a finite field (if the field has order $q$ then the group
has order $q(q-1)$) with probability $1-\varepsilon$ with (polylog) complexity
$O(\log(q)^{R}\log(\varepsilon)^{2})$ where $R<\infty$ is independent of $q$
(more detail on the size of $R$ can be found in sections 5 and 6).There has
been a great deal of research on this problem with emphasis on the case when
$q$ is a prime (see most notably [M,R,R,S]). Here $q$ is general (that is, a
power of a prime). To us the most interesting case is that of $q=2^{n}$.

Our method involves the fact that the group has exactly one irreducible
representation that is not one dimensional (its dimension is $q-1$and was also
a key player in [M,R,R,S]). Associated with this representation is a wavelet
transform. These wavelets are the finite analogs of the original wavelets of
Morlet and Grossmann [GM]. In this paper we use that fact that the wavelets
are determined up to phase as the unique unit fixed points of the subgroups
for which we are searching. In the first two sections we give a fairly self
contained exposition of this theory. The other ingredient is elementary number
theory. We use the basic theory of Gauss sums and give enough detail to
explain the part that is pertinent (c.f. [Ros]).

We thank David Meyer for his help and encouragement in the evolution of the
results of this paper. Especially for the lessons he gave to this novice in
quantum mechanics. This work began seven years ago while Meyer and the author
were partially supported by a Darpa grant to study quantum wavelets.

\section{The $ax+b$ groups and their representations}

Let $F$ be a finite field and let $F^{\times}=F-\{0\}$. We will consider the
group%
\[
G=\left\{  \left[
\begin{array}
[c]{cc}%
a & b\\
0 & 1
\end{array}
\right]  |a\in F^{\times},b\in F\right\}  .
\]
this is the group of all affine transformations of $F$. As a set $G$ is
$F^{\times}\times F$. Let $q$ be the cardinality of $F$ then $q=p^{n}$ with
$p$ a prime. Also $F^{\times}$ is a cyclic group of order $p^{n}-1$. Fix a
cyclic element $u\in F^{\times}$ and write $\zeta=e^{\frac{2\pi i}{q-1}}$.
Then $F^{\times}$ has $q-1$ one dimensional characters, $\chi_{x}$, $x\in
F^{\times}$, defined by $\chi_{u^{j}}(u^{l})=\zeta^{jl}$. We will also use the
notation $\chi_{j}$ for $\chi_{u^{j}}$. We extend these to characters of $G$
by%
\[
\chi_{x}\left(  \left[
\begin{array}
[c]{cc}%
a & b\\
0 & 1
\end{array}
\right]  \right)  =\chi_{x}(a).
\]
In addition, $G$ has an irreducible representation of dimension $q-1$. which
we will now describe. Let $L^{2}(F)$ be the space of all functions from $F$ to
$%
\mathbb{C}
$ with inner product%
\[
\left\langle f|g\right\rangle =\sum_{x\in F}\overline{f(x)}g(x).
\]
We define
\[
\pi\left(  \left[
\begin{array}
[c]{cc}%
a & b\\
0 & 1
\end{array}
\right]  \right)  f(x)=f\left(  \frac{x-b}{a}\right)  .
\]
Then $(\pi,L^{2}(F))$ defines a unitary representation of $G$. The constant
functions form an invariant subspace. Their orthogonal complement, $L_{0}%
^{2}(F)$, is the space of all functions $f\in L^{2}(F)$ such that
\[
\sum_{x\in G}f(x)=0.
\]
This representation is irreducible and since $(q-1)^{2}+(q-1)=q(q-1)$, which
is the order of $G$, we see that we have described up to equivalence all
irreducible unitary representations of $G$. We write $\pi_{0}(g)$ for
$\pi(g)_{|L_{0}^{2}(F)}.$We will take $L^{2}(F^{\times})$ to be the subspace
of $L^{2}(F)$ orthogonal to the delta function at $0$.

We take as our Hilbert space $L^{2}(G)$ and as the computational basis the set
of delta functions $\left\vert g\right\rangle $ for $g\in G$. Since $G$ is set
theoretically $F^{\times}\times F$ we have $L^{2}(G)$ is the tensor product of
$L^{2}(F^{\times})$ and $L^{2}(F)$. We take in both cases the computational
basis to be the delta functions. That is $\left\vert x\right\rangle
=\delta_{x}$ and $\delta_{x}(y)=1$ if $x=y$ and $0$ otherwise. We take our
computational basis to be the elements%
\[
\left\vert \left[
\begin{array}
[c]{cc}%
u^{j} & x\\
0 & 1
\end{array}
\right]  \right\rangle =\left\vert u^{j},x\right\rangle =\left\vert
u^{j}\right\rangle \otimes\left\vert x\right\rangle ,j=0,...,q-2,x\in F.
\]
This implies that the left regular representation, $(L,L^{2}(G))$ of $G$ is
given by $\pi_{1}(x)\otimes\pi(x)$ on $L^{2}(F^{\times})\otimes L^{2}(F)$ with
$\pi(x)=\pi\left(  \left[
\begin{array}
[c]{cc}%
\alpha & \beta\\
0 & 1
\end{array}
\right]  \right)  $ as above and$\pi_{1}(\alpha)f(z)=f(\alpha^{-1}z)$ for
$f\in L^{2}(F^{\times})$. To prove this simple fact we note that if
\[
x=\left[
\begin{array}
[c]{cc}%
\alpha & \beta\\
0 & 1
\end{array}
\right]  ,\text{ }%
\]
then%
\[
\text{ }x^{-1}=\left[
\begin{array}
[c]{cc}%
\alpha^{-1} & -\alpha^{-1}\beta\\
0 & 1
\end{array}
\right]
\]
and as usual
\[
L_{x}f(y)=f(x^{-1}y).
\]
So if
\[
y=\left[
\begin{array}
[c]{cc}%
a & b\\
0 & 1
\end{array}
\right]
\]
then%
\[
L_{x}f(y)=f\left(  \left[
\begin{array}
[c]{cc}%
\alpha^{-1}a & \frac{b-\beta}{\alpha}\\
0 & 1
\end{array}
\right]  \right)  .
\]
We also note that
\[
\pi\left(  \left[
\begin{array}
[c]{cc}%
a & b\\
0 & 1
\end{array}
\right]  \right)  \left\vert x\right\rangle =\left\vert ax+b\right\rangle .
\]

\section{A class of maximal subgroups and the corresponding wavelets}

For each $b\in F$ we consider the cyclic subgroup $C_{b}$ of order $q-1$
consisting of the powers of
\[
\left[
\begin{array}
[c]{cc}%
u & b\\
0 & 1
\end{array}
\right]  =\left\vert u,b\right\rangle .
\]
We want an efficient algorithm for the hidden subgroup problem involving these
groups. That is, we have a set $X$ with $q$ elements (say $X=F)$. \ We assume
that we have a function $f:G\rightarrow X$ that has the property that
$f(cg)=f(g)$ for $c\in C_{b}$ and $f$ takes exactly $q$ values. Our problem is
to determine $b$. We note that%
\[
\left[
\begin{array}
[c]{cc}%
u & b\\
0 & 1
\end{array}
\right]  ^{k}=\left[
\begin{array}
[c]{cc}%
u^{k} & (1+u+...+u^{k-1})b\\
0 & 1
\end{array}
\right]  =\left[
\begin{array}
[c]{cc}%
u^{k} & \frac{1-u^{k}}{1-u}b\\
0 & 1
\end{array}
\right]  .
\]
This implies

\begin{lemma}
If $b\neq0$ then if $\left\vert x,y\right\rangle ,\left\vert x^{\prime
},y^{\prime}\right\rangle \in C_{b}$ and $y=y^{\prime}$ then $x=x^{\prime}$.
\end{lemma}

We also note that all of these groups are conjugate. Indeed%
\[
C_{b}=\left[
\begin{array}
[c]{cc}%
1 & \frac{b}{u-1}\\
0 & 1
\end{array}
\right]  C_{0}\left[
\begin{array}
[c]{cc}%
1 & \frac{-b}{u-1}\\
0 & 1
\end{array}
\right]  .
\]
This implies that%
\[
\pi\left(  \left[
\begin{array}
[c]{cc}%
u & b\\
0 & 1
\end{array}
\right]  \right)  f(x)=f(u^{-1}(x-\frac{b}{1-u})+\frac{b}{1-u}).
\]

We note that if $k\neq0$ then if $\phi\in L^{2}(F)$ and
\[
\pi\left(  \left[
\begin{array}
[c]{cc}%
u & 0\\
0 & 1
\end{array}
\right]  \right)  \phi=\zeta^{k}\phi
\]
then up to scalar $\phi$ is equal to%
\[
\phi_{k.0}(x)=\left\{
\begin{array}
[c]{c}%
\frac{1}{\sqrt{q-1}}\zeta^{-kj}\text{ if }x=u^{j}\\
0\text{ if }x=0
\end{array}
\right.  .
\]
We note that if $k=0$ then $\phi_{0,0}$ is invariant but so is $\left\vert
0\right\rangle $. In all cases we set for $v\in F,$ $\phi_{k,v}(x)=\phi
_{k,0}(x-v)$ =%
\[
\phi_{k,v}(x)=\phi_{k,0}(x-v)=\pi\left(  \left[
\begin{array}
[c]{cc}%
1 & v\\
0 & 1
\end{array}
\right]  \right)  \phi_{k,0}(x)
\]
We also note that if $k\neq0$ then $\phi_{k,v}\in L_{0}^{2}(F)$. If $k=0$ we
set
\[
\psi_{0}=\frac{\phi_{0,0}-\sqrt{q-1}\left\vert 0\right\rangle }{\sqrt{q}}\in
L_{0}^{2}(F).
\]
We set%
\[
\psi_{v}=\pi_{0}\left(  \left[
\begin{array}
[c]{cc}%
1 & v\\
0 & 1
\end{array}
\right]  \right)  \psi_{0}%
\]
then the calculations above imply that if $c=\left[
\begin{array}
[c]{cc}%
u & b\\
0 & 1
\end{array}
\right]  $ then

\begin{lemma}
$\pi_{0}(c)\psi_{\frac{b}{1-u}}=\psi_{\frac{b}{1-u}}$ and $\pi_{0}%
(c)\phi_{k,\frac{b}{1-u}}=\zeta^{k}\phi_{k,\frac{b}{1-u}}$ and if $\phi\in
L_{0}^{2}(F)$ and $\pi_{0}(c)\phi=\phi$ (resp.$\pi_{0}(c)\phi=\zeta^{k}\phi$
with $k\neq0)$ then $\phi=z\psi_{\frac{b}{1-u}}$ with $z\in%
\mathbb{C}
$ (resp. $\phi=z\phi_{k,\frac{b}{1-u}}$). Furthermore, if we measure $\psi
_{v}$ in the computational basis it collapses to $\left\vert v\right\rangle $
with probability $\frac{q-1}{q}$.
\end{lemma}

\begin{lemma}
If $H\subset G$ is a maximal subgroup that is non-normal and cyclic then
$H=C_{b}$ for some $b\in F$.
\end{lemma}

\begin{proof}
Let
\[
\left[
\begin{array}
[c]{cc}%
u^{j} & x\\
0 & 1
\end{array}
\right]
\]
be a generator then if $u^{j}=1$ the group is maximal if and only if $x$
generates the additive group of $F$. So if maximal $H$ is normal. Otherwise,
$u^{j}\neq1$ and thus
\[
H\subset\left[
\begin{array}
[c]{cc}%
1 & \frac{x}{u^{j}-1}\\
0 & 1
\end{array}
\right]  C_{0}\left[
\begin{array}
[c]{cc}%
1 & \frac{-x}{u^{j}-1}\\
0 & 1
\end{array}
\right]  .
\]
Hence $H$ is maximal only if the inclusion is an equality. But then $H=C_{b}$
with $b=\frac{x(u-1)}{u^{j}-1}$.
\end{proof}

\section{Preliminary considerations}

We begin by considering the state%
\[
v=\frac{1}{\sqrt{q-1}}\sum_{c\in C_{b}}\left\vert cg\right\rangle
\]
for some $g\in G,$ $b\in F$. \ We first observe that $v$ is invariant under
the left regular action of $C_{b}$.

Let $\mathcal{F}_{M}$ denote the quantum Fourier transform for the cyclic
multiplicative group $F^{\times}$. Recall that
\[
\mathcal{F}_{M}\left\vert u^{j}\right\rangle =\frac{1}{\sqrt{q-1}}\sum
_{k=0}^{q-2}\zeta^{kj}\left\vert u^{k}\right\rangle .
\]
Now assume that
\[
g=\left[
\begin{array}
[c]{cc}%
u^{r} & x\\
0 & 1
\end{array}
\right]  ,c=\left[
\begin{array}
[c]{cc}%
u^{j} & \frac{u^{j}-1}{u-1}b\\
0 & 1
\end{array}
\right]
\]
then%
\[
\left\vert cg\right\rangle =\left[
\begin{array}
[c]{cc}%
u^{j+r} & \frac{u^{j}-1}{u-1}(b+(u-1)x)+x\\
0 & 1
\end{array}
\right]  .
\]

We consider $w=\left(  \mathcal{F}_{M}\otimes I\right)  v$ that is%
\[
w=\frac{1}{q-1}\sum_{j,k=0}^{q-2}\zeta^{k(j+r)}\left\vert u^{k}\right\rangle
\otimes\left\vert \frac{u^{j}-1}{u-1}(b+(u-1)x)+x\right\rangle .
\]
There are two cases

a) $x=\frac{b}{1-u}$ then
\[
w=\frac{1}{q-1}\sum_{j,k=0}^{q-2}\zeta^{k(j+r)}\left\vert u^{k}\right\rangle
\otimes\left\vert \frac{b}{1-u}\right\rangle =\left\vert 1\right\rangle
\otimes\left\vert \frac{b}{1-u}\right\rangle .
\]

b) $x\neq\frac{b}{1-u}$ then the elements $\left\vert \frac{u^{j}-1}%
{u-1}(b+(u-1)x)+x\right\rangle $ for $j=0,...,q-2$ are all distinct. The
coefficient of $\left\vert 0\right\rangle $ is
\[
\frac{1}{q-1}\sum_{j=0}^{q-2}\left\vert \frac{u^{j}-1}{u-1}b+u^{j}%
x\right\rangle =\frac{1}{q-1}\left(  \sum_{j=0}^{q-1}\left\vert j\right\rangle
-\left\vert \frac{b}{1-u}\right\rangle \right)  .
\]
The coefficient of $\left\vert u^{k}\right\rangle \neq\left\vert
1\right\rangle $ is
\[
\frac{\zeta^{kr}}{q-1}\sum_{j=0}^{q-2}\zeta^{kj}\left\vert \frac{u^{j}-1}%
{u-1}b+u^{j}x\right\rangle .
\]
We note that this element is in $L_{0}^{2}(F)$. We also note that $\pi\left(
\left[
\begin{array}
[c]{cc}%
\alpha & \beta\\
0 & 1
\end{array}
\right]  \right)  \left\vert z\right\rangle =\left\vert \alpha z+\beta
\right\rangle .$This implies that%
\[
\pi\left(  \left[
\begin{array}
[c]{cc}%
u & b\\
0 & 1
\end{array}
\right]  \right)  \sum_{j=0}^{q-2}\zeta^{kj}\left\vert \frac{u^{j}-1}%
{u-1}b+u^{j}x\right\rangle =
\]%
\[
\sum_{j=0}^{q-2}\zeta^{kj}\left\vert u\frac{u^{j}-1}{u-1}b+u^{j+1}%
x+b\right\rangle =\sum_{j=0}^{q-2}\zeta^{kj}\left\vert \frac{u^{j+1}-1}%
{u-1}b+u^{j+1}x\right\rangle
\]%
\[
=\zeta^{-k}\sum_{j=0}^{q-2}\zeta^{kj}\left\vert \frac{u^{j}-1}{u-1}%
b+u^{j}x\right\rangle .
\]
Combining these calculations and the Lemma 2 we have

\begin{theorem}
If $x=\frac{b}{1-u}$ then $\left(  \mathcal{F}_{M}\otimes I\right)  v=$
$\left\vert 1\right\rangle \otimes\left\vert \frac{b}{1-u}\right\rangle $
otherwise
\[
\left(  \mathcal{F}_{M}\otimes I\right)  v=\frac{1}{q-1}\left\vert
1\right\rangle \otimes\left(  \sum_{j=0}^{q-1}\left\vert j\right\rangle
-\left\vert \frac{b}{1-u}\right\rangle \right)  +\sum_{k\neq0}\left\vert
u^{k}\right\rangle \otimes f_{k}%
\]
with $f_{k}$ for $k\neq0$ a non-zero multiple of $\phi_{-k,\frac{b}{1-u}}$.
\end{theorem}

\bigskip

\section{The method}

We now assume that we have a surjective function $f:G\rightarrow
\{0,1,...,q-1\}$(which we look upon as $%
\mathbb{Z}
/q%
\mathbb{Z}
$) that is constant on the cosets of $C_{b}$ with $b\in F$. We will now
describe a method of determining $b$. We begin with the initial state%
\[
\frac{1}{\sqrt{q(q-1)}}%
{\displaystyle\sum\limits_{g\in G}}
\left\vert g\right\rangle \otimes\left\vert 0\right\rangle \in L^{2}(G)\otimes%
\mathbb{C}
\lbrack%
\mathbb{Z}
/q%
\mathbb{Z}
].
\]

We apply the unitary transformation on $L^{2}(G)\otimes%
\mathbb{C}
\lbrack%
\mathbb{Z}
/q%
\mathbb{Z}
]$ defined by%
\[
\left\vert g\right\rangle \otimes\left\vert x\right\rangle \longmapsto
\left\vert g\right\rangle \otimes\left\vert x+f(g)\right\rangle
\]
on the first two tensor factors. Obtaining%
\[
\frac{1}{\sqrt{q(q-1)}}%
{\displaystyle\sum\limits_{g\in G}}
\left\vert g\right\rangle \otimes\left\vert f(g)\right\rangle .
\]
We measure the second tensor factor (and discard it) and the state collapses
to%
\[
\frac{1}{\sqrt{q-1}}%
{\displaystyle\sum\limits_{c\in C_{b}}}
\left\vert cg\right\rangle
\]
With $g$ a random element of $G$. We write it as%
\[
\left[
\begin{array}
[c]{cc}%
u^{r} & x\\
0 & 1
\end{array}
\right]  .
\]
Since the probability of the random element $x$ having the value $\frac
{b}{1-u}$ is $\frac{1}{q}$ we will assume that $x\neq\frac{b}{1-u}$. Thus in
the notation of the previous sections we have%
\[
\frac{1}{\sqrt{q-1}}%
{\displaystyle\sum\limits_{c\in C_{b}}}
\left\vert u^{j+r}\right\rangle \otimes\left\vert u^{j}x+\frac{u^{j}-1}%
{u-1}b\right\rangle .
\]
We now apply $\mathcal{F}_{M}\otimes I$ which according to Proposition 4
yields (in the notation of that result)
\[
\frac{1}{q-1}\left\vert 1\right\rangle \otimes\left(  \sum_{j=0}%
^{q-1}\left\vert j\right\rangle -\left\vert \frac{b}{1-u}\right\rangle
\right)  +\sum_{k=0}^{q-2}\left\vert u^{k}\right\rangle \otimes f_{k}.
\]
We note that the two terms are orthogonal thus if we measure the first factor
then the probability of measuring some $\left\vert k\right\rangle \neq0$ is
$\frac{q-2}{q-1}$ so we may assume that we have measured some $k\neq0$. We
discard the factor $\left\vert k\right\rangle $ and now have up to phase
$\phi_{-k,\frac{b}{1-u}}$. We now take an additive quantum Fourier transform
and have $\mathcal{F}_{A}\phi_{-k,\frac{b}{1-u}}$.

We are thus left with finding $b$ knowing $k$ and $\mathcal{F}_{A}%
\phi_{-k,\frac{b}{1-u}}$. \ For simplicity of notation we will replace
$\frac{b}{1-u}$ with $b$. We will now calculate $\mathcal{F}_{A}\phi_{k,b}$ in
$L^{2}(F)$. First
\[
\phi_{-k,b}=\frac{1}{\sqrt{q-1}}\sum_{y\in F^{\times}}\chi_{k}(y)\left\vert
y+b\right\rangle .
\]

We now calculate the (additive) quantum Fourier transform of this state
$\mathcal{F}_{A}(\phi_{-k,b})$. Recalling that it is defined as follows. Fix a
non-trivial character of the additive group $F,$ $\eta^{-1}$. Then%
\[
\mathcal{F}_{A}\left\vert z\right\rangle =\frac{1}{\sqrt{q}}\sum_{t\in F}%
\eta(tz)\left\vert t\right\rangle .
\]

Thus%
\[
\mathcal{F}_{A}\phi_{k,b}=\frac{1}{\sqrt{q(q-1)}}\sum_{%
\begin{array}
[c]{c}%
y\in F^{\times}\\
t\in F
\end{array}
}\chi_{k}(y)\eta(t(y+b))\left\vert t\right\rangle =
\]%
\[
\frac{1}{\sqrt{q(q-1)}}\sum_{t\in F}\left(  \sum_{y\in F^{\times}}\chi
_{k}(y)\eta(ty)\right)  \eta(tb)\left\vert t\right\rangle .
\]
We note that if $t=0$ then $\sum_{y\in F^{\times}}\chi_{k}(y)\eta
(ty)=\sum_{y\in F^{\times}}\chi_{k}(y)=0$ since $k\neq0$. If $z\in F^{\times}
$ and $z=u^{j}$ then we set for $t\in F$%
\[
G(u^{j},t)=\frac{1}{\sqrt{q}}\sum_{y\in F^{\times}}\chi_{j}(y)\eta(ty).
\]
This (except for normalization) is known to number theorists as a Gauss sum.
We have shown%
\[
\mathcal{F}_{A}\phi_{k,b}=\frac{1}{\sqrt{q-1}}\sum_{t\in F^{\times}}%
G(u^{k},t)\eta(tb)\left\vert t\right\rangle .
\]
The following is a standard fact [c.f. [Ros]) about Gauss sums (the following
lemma is standard we include a proof for the convenience of the reader).

\begin{lemma}
If $z,t\in F^{\times}$ then
\[
G(u^{j},z)\overline{G(u^{j},t)}=\chi_{j}(\frac{z}{t}).
\]
In particular, $\left\vert G(u^{j},z)\right\vert =1$ and $\overline
{G(u^{j},1)}G(u^{j},z)=\chi_{j}(z)$.
\end{lemma}

\begin{proof}
We calculate%
\[
qG(u^{j},z)\overline{G(u^{j},t)}=\sum_{y\in F^{\times}}\chi_{j}(y)\eta
(zy)\sum_{y\in F^{\times}}\chi_{j}(y^{-1})\eta(-ty)=
\]%
\[
\sum_{y,s\in F^{\times}}\chi_{j}(ys^{-1})\eta(zy-ts).
\]
In this expression we change variables $yt^{-1}\rightarrow y$ and get%
\[
\sum_{y,s\in F^{\times}}\chi_{j}(y)\eta(zys-ts)=\sum_{y,s\in F^{\times}}%
\chi_{j}(y)\eta((zy-t)s).
\]
Fixing $y$ and summing in $s$ we see that if $zy\neq t$ then the sum in $s$
is
\[
\chi_{j}(y)\sum_{s\neq0}\eta(s)=-\chi_{j}(y).
\]
If $y=\frac{z}{t}$ then the sum in $t$ is $(q-1)\chi_{j}(\frac{z}{t}).$ This
implies that the total sum is%
\[
(q-1)\chi_{j}(\frac{z}{t})-\sum_{y\neq\frac{z}{t}}\chi_{j}(y)=(q-1)\chi
_{j}(\frac{z}{t})+\chi_{j}(\frac{z}{t}).
\]

\end{proof}

Since states are determined up to phase if we multiply $\mathcal{F}_{A}%
\phi_{k,b}$ by $\overline{G(u^{k},1)}$ we have
\[
\mathcal{F}_{A}\phi_{k,b}=\frac{1}{\sqrt{q-1}}\sum_{s\in F^{\times}}\chi
_{k}(s)\eta(sb)\left\vert s\right\rangle .
\]
Consider the unitary operators, on $L^{2}(F)$ given by $T_{j}\left\vert
t\right\rangle =\chi_{j}(t)^{-1}\left\vert t\right\rangle ,$ $t\neq0$ and
$T_{j}\left\vert 0\right\rangle =\left\vert 0\right\rangle $. For the moment
we assume that we have implemented these operators one state at a time.
Applying $T_{k}$ we are left with calculating $b$ given%
\[
\frac{1}{\sqrt{q-1}}\sum_{s\in F^{\times}}\eta(sb)\left\vert s\right\rangle .
\]
If we apply the inverse of the additive quantum Fourier transform to this we
have%
\[
\mathcal{F}_{A}^{-1}\frac{1}{\sqrt{q-1}}\sum_{s\in F^{\times}}\eta
(sb)\left\vert s\right\rangle =\psi_{b}.
\]
Measuring this in the computational basis of $L^{2}(F)$ yields a collapse to
$\left\vert b\right\rangle $ with probability $\frac{q-1}{q}$.

Thus with probability $\left(  \frac{q-1}{q}\right)  ^{2}$ we can calculate
$b$ with complexity at most the cost of iteration one multiplicative Quantum
Fourier transform and 3 additive Quantum Fourier transforms and one value of
instance of $T_{j}$. In the next section we will give in the next section a
quantum algorithm that calculates $T_{j}$ with the complexity of two
multiplicative Quantum Fourier transforms and $O(-\log\log(q-1)\log
(\varepsilon))$ to get a probability of success $1-\varepsilon$.

\section{The last step}

In the last section we found an efficient method of solving the hidden
subgroup algorithm with probability $\left(  \frac{q-1}{q}\right)  ^{2}$ of
success assuming that we have an oracle that calculates the unitary operator
on $L^{2}(F)$ defined by%
\[
T_{k}\left\vert x\right\rangle =\left\{
\begin{tabular}
[c]{l}%
$\left\vert 0\right\rangle $ if $x=0$\\
$\zeta^{-jk}\left\vert x\right\rangle $ if $x=\left\vert u^{j}\right\rangle $%
\end{tabular}
\ \right.  .
\]

We note that if $x\neq0$ and $x=u^{j}$ with $0\leq j\leq q-2$ then $j$ is the
discrete log of $x$ and we denote it by $L(x).$

Given a state $v=\sum_{x\in F^{\times}}a_{x}\left\vert x\right\rangle $ we
will calculate $T_{k}v$ with complexity two multiplicative Fourier transforms,
two superpositions of polylog operations times $O((\log\log q)(\log
q)(-\log\varepsilon)$ and probability of success $1-\varepsilon$. . We begin
as in the implementation of Shor's discrete log algorithm in [NC, p.238] with
the state
\[
\frac{1}{q-1}\sum_{r,s=0}^{q-2}\sum_{x\in F-\{0\}}a_{x}\left\vert
r\right\rangle \otimes\left\vert x\right\rangle \otimes\left\vert
s\right\rangle \otimes\left\vert 0\right\rangle
\]
with the second and third tensor factor in $F^{\times}$ and the last factor is
in $L^{2}(F)$. Next as in [NC] define the unitary operator
\[
V:%
\mathbb{C}
^{q-1}\otimes L^{2}(F)\otimes%
\mathbb{C}
^{q-1}\otimes L^{2}(F)\rightarrow%
\mathbb{C}
^{q-1}\otimes L^{2}(F)\otimes%
\mathbb{C}
^{q-1}\otimes L^{2}(F)
\]
by%
\[
V\left\vert r\right\rangle \otimes\left\vert x\right\rangle \otimes\left\vert
s\right\rangle \otimes\left\vert z\right\rangle =\left\vert r\right\rangle
\otimes\left\vert x\right\rangle \otimes\left\vert s\right\rangle
\otimes\left\vert z+x^{r}u^{s}\right\rangle .
\]
If we apply $V$ to our initial state we have
\[
\frac{1}{q-1}\sum_{r,s=0}^{q-1}\sum_{x\neq0}a_{x}\left\vert r\right\rangle
\otimes\left\vert x\right\rangle \otimes\left\vert s\right\rangle
\otimes\left\vert x^{r}u^{s}\right\rangle .
\]

We observe that if $x\neq0$, $x=u^{L(x)}$ with $0\leq L(x)<q-1$ (the discrete
log). Thus the expression we obtain is%
\[
\frac{1}{q-1}\sum_{r,s=0}^{q-2}\sum_{x\neq0}a_{x}\left\vert r\right\rangle
\otimes\left\vert x\right\rangle \otimes\left\vert s\right\rangle
\otimes\left\vert u^{L(x)r+s}\right\rangle .
\]
To this we apply $\mathcal{F}_{M}\otimes I\otimes\mathcal{F}_{M}\otimes I$ and
obtain%
\[
\frac{1}{(q-1)^{2}}\sum_{r,s,l,m=0}^{q-2}\sum_{x\neq0}\zeta^{lr+ms}%
a_{x}\left\vert x\right\rangle \otimes\left\vert l\right\rangle \otimes
\left\vert m\right\rangle \otimes\left\vert u^{L(x)r+s}\right\rangle .
\]
We make the change index of summation $s\rightarrow z=L(x)r+s$ so the sum
becomes%
\[
\frac{1}{(q-1)^{2}}\sum_{x\neq0}a_{x}\sum_{r,z,l,m=0}^{q-1}\zeta
^{(l-mL(x))r+mz}\left\vert l\right\rangle \otimes\left\vert x\right\rangle
\otimes\left\vert m\right\rangle \otimes\left\vert u^{z}\right\rangle .
\]
The sum over $r$ is $0$ if $l\neq mL(x)$ thus we obtain%
\[
\frac{1}{q-1}\sum_{x\neq0}a_{x}\sum_{z,m=0}^{q-1}\zeta^{mz}\left\vert
L(x)m\right\rangle \otimes\left\vert x\right\rangle \otimes\left\vert
m\right\rangle \otimes\left\vert u^{z}\right\rangle .
\]
We now measure the third and fourth tensor factors discard and obtain for some
$m$ random but now known
\[
\sum_{x\neq0}a_{x}\left\vert L(x)m\right\rangle \otimes\left\vert
x\right\rangle .
\]
If $m$ is relatively prime to $q-1$ then we can use the Euclidean algorithm to
efficiently calculate the inverse of $m$ in $%
\mathbb{Z}
/(q-1)%
\mathbb{Z}
$ getting%
\[
\sum_{x\neq0}a_{x}\left\vert L(x)\right\rangle \otimes\left\vert
x\right\rangle
\]
With probability at least
\[
\frac{\phi(q-1)}{q-1}%
\]
and complexity at most that of the discrete logarithm of size $q-1$ and the
classical complexity of taking powers of $\zeta$.

This is classically the function, $h,$ on $\mathbb{Z}/(q-1)\mathbb{Z}\times
F$, that assigns to the pair $(L(x),x)$ the value $a_{x}$ and to all other
elements $0$. We also define classically the function $U_{k}(j,x)=\zeta^{-kj}%
$. Then our desired function is $U_{k}h$. Thus we have using polynomial time
classical operations
\[
u=\sum_{x\neq0}\zeta^{-kL(x)}a_{x}\left\vert L(x)\right\rangle \otimes
\left\vert x\right\rangle .
\]
We consider the unitary operator, $T$ , From $L^{2}(\mathbb{Z}/(q-1)\mathbb{Z}%
\times F)$ to $L^{2}(F^{\times}\times F)$ defined by%
\[
T\left\vert j,x\right\rangle =\left\vert u^{j},u^{-j}x\right\rangle
\]
which is a superposition of classically logpoly operations. This yields
\[
Tu=\sum_{x\neq0}\zeta^{-kL(x)}a_{x}\left\vert u^{L(x)}\right\rangle
\otimes\left\vert 1\right\rangle .
\]
We discard the factor $\left\vert 1\right\rangle $ and get%
\[
\sum_{x\neq0}\zeta^{-kL(x)}a_{x}\left\vert x\right\rangle .
\]
This is our desired implementation of $T_{k}$. Applying it to the state with
\[
a_{x}=\frac{\zeta^{kL(x)}\eta(xb)}{\sqrt{q-1}}%
\]
completes the method described in the last section..

\begin{lemma}
If we apply the above algorithm $-2\log\log(q-1)\log(\varepsilon)$ times then
the probability of success is $1-\varepsilon$ (here $\log$ is the natural logarithm).
\end{lemma}

\begin{proof}
Set $\phi(n)$ equal to the number integers $m$ with $1<m<$ $n$ and
$\gcd(m,n)=1$. One has (cf. [HW]
\[
\frac{\phi(n)}{n}\geq e^{-\gamma}\frac{1}{(\log\log n+\frac{3}{\log\log n}%
)}>\frac{1}{2\log\log n}.
\]
With $\gamma$ equal to Euler's constant so $e^{-\gamma}>0.56$. Thus the
probability of failure of the above algorithm after $2\log\log(q-1)$ runs is
less than
\[
\left(  1-\frac{1}{2\log\log(q-1)}\right)  ^{2\log\log(q-1)}.
\]
Using elementary calculus we see that $\left(  1-\frac{1}{x}\right)  ^{x}$ is
monotonically increasing to $\frac{1}{e}$ for $x>1$. This implies that after
$2\log\log(q-1)$ steps the probability of failure is $<\frac{1}{e}$. If we run
tests $-\log(\varepsilon)$ times then the probability of failure is at most
$\varepsilon$. This proves the result.
\end{proof}

We have completed the proof of the result asserted in the introduction.

We note that for certain infinite sequences of fields of order a power of 2 we
can leave off the loglog term. That is find infinite sequences of choices of
$q$ with $\frac{\phi(q-1)}{q-1}$ bounded below. We will carry this out in the
next section.

We now calculate the complexity of the algorithm. The number of multiplicative
Fourier transforms used is $4$ and the number of additive ones is $2$. The
last step involves a subroutine that has complexity the preparation of an
initial state and the classical operations described above.

\section{Examples}

\begin{enumerate}
\item We fix $p$ a prime. The sequence of values of $q$ will be $p^{n}$ with
$n$ a prime. We assert that for such $q$ we have%
\[
\frac{\phi(q-1)}{q-1}\geq C_{p}%
{\displaystyle\prod\limits_{%
\begin{array}
[c]{c}%
l\text{ prime}\\
l<p
\end{array}
}}
\left(  1-\frac{1}{l}\right)
\]
where $C_{p}$ is the minimum of $\left(  1-\frac{1}{2n+1}\right)  ^{n}$ for
$n$ a prime $n>p$. (This number for $p$ large is arbitrarily close to
$e^{\frac{-1}{2}}$.) To see this we observe that if $l$ is a prime that
divides $p^{n}-1$ then
\[
p^{n}\equiv1\text{ mod }l
\]
if $l\neq2$ then Fermat's little theorem implies that $n$ divides $l-1$. Thus
$l=kn+1$. Assuming that $n\neq2$ this implies that $k$ is even. Now we have%
\[
p^{n}-1=%
{\displaystyle\prod\limits_{%
\begin{array}
[c]{c}%
l\text{ prime}\\
l|(p^{n}-1)
\end{array}
}}
l^{e_{l}}.
\]
Thus
\[
n>\text{log}_{p}(p^{n}-1)=\sum_{l}e_{l}\text{log}_{p}l.
\]
This implies that the number of primes $l>p$ that divide $p^{n}-1$ is less
then $n$. \ Now
\[
\frac{\phi(p^{n}-1)}{p^{n}-1}=%
{\displaystyle\prod\limits_{%
\begin{array}
[c]{c}%
l\text{ prime}\\
l|(p^{n}-1)
\end{array}
}}
\left(  1-\frac{1}{l}\right)  .
\]
Applying the observations above about such primes we have%
\[
\frac{\phi(p^{n}-1)}{p^{n}-1}\geq\left(  1-\frac{1}{2n+1}\right)  ^{n}%
{\displaystyle\prod\limits_{%
\begin{array}
[c]{c}%
l\text{ prime}\\
l<p
\end{array}
}}
\left(  1-\frac{1}{l}\right)  .
\]

\item For special values of $p$ we can do much better. If $p=2$ we can by
doing computer calculations see that a lower bound of $.6$ prevails. If $p=3$
it computer calculation suggests that a lower bound of $.45$ works. We can
prove (with the help of a computer) $.3$. The case of $p=2$ is especially
interesting since there are extremely large primes, $n$, such that $2^{n}-1$
is a prime (Mersenne prime). The record at this time is $n=43,112,609$. If
$q-1$ is a prime then $\frac{\phi(q-1)}{q-1}=1-\frac{1}{q-1}.$

\item Another class of examples that are especially important to discrete
wavelets is the $q=2^{2^{n}}$. \ Computer experimentation suggests that there
is a lower bound of $.4997$. We can prove that there is a positive lower bound
independent of $n$ as follows. Let $F_{n}=2^{2^{n}}+1$ (the $n$-th Fermat
number). It is well known that these numbers are relatively prime. Also it is
easy to see that $2^{2^{n}}-1=F_{0}F_{1}\cdots F_{n-1}$. Thus if $\mathcal{S}$
is the set of primes that divide some Fermat number then%
\[
\frac{\phi(2^{2^{n}}-1)}{2^{2^{n}}-1}\geq%
{\displaystyle\prod\limits_{l\in\mathcal{S}}}
\left(  1-\frac{1}{l}\right)  .
\]
One can prove that this product is bounded below using the fact ([KLS)%
\[
\sum_{l\in\mathcal{S}}\frac{1}{l}<\infty\text{.}%
\]
We should note that Fermat made the (incorrect) conjecture that $F_{n}$ is
always prime. If that were so then the product above would be
\[%
{\displaystyle\prod\limits_{n=0}^{\infty}}
\left(  1-\frac{1}{2^{2^{n}}+1}\right)  =\frac{1}{2}.
\]
The computed value for the above product is $.4997...$ so Fermat was almost right.
\end{enumerate}

\pagebreak

\begin{center}
{\LARGE References}
\end{center}

\noindent\lbrack GM]A.Grossmann and J. Morlet, Decomposition of Hardy
functions into square integrable wavelets of constant shape., SIAM J. Math.
Anal.,15 (1984),723-736

\noindent\lbrack KLS] Michal K\v{r}\'{\i}\v{z}ek,Florian Luca and Lawrence
Somer, On the convergence of series of reciprocals of primes related to the
Fermat numbers.(English summary) J. Number Theory 97 (2002), no. 1,
95--112.\smallskip

\noindent\lbrack M,R,R,S] C. Moore, D. Rockmore, A. Russell,L. Schullman, The
Power of Strong Fourier Sampling: Quantum Algorithms for Affine Groups and
Hidden Shifts, SIAMJournal on Computingf 37 (2007) 938-958,  SODA '04,
1113-1122.

\noindent\lbrack NC] Michael Nielsen and Isaac Chuang, Quantum Computation and
Quantum Information, Cambridge University Press, 2000.\smallskip

\noindent\lbrack Ro] H. E. Rose, A course in number theory,Second edition,
Oxford Science Publications,Oxford, 1994.\smallskip

\noindent\lbrack R] Barkley Rosser, Explicit bounds for some functions of
prime numbers, Amer. J. Math, 63(1941),211-232.

\noindent\lbrack HW] G. H. Hardy and E. M. Wright, An Introduction to the
Theory of Numbers,Fifth Edition,Oxford Science Publications,Oxford,
1979.\smallskip

\end{document}